\documentclass[twoside,11pt]{article}
\usepackage{jair}

\usepackage{times}

\usepackage{amsmath}
\usepackage{amssymb}
\usepackage{amsthm}
\usepackage[usenames,dvipsnames]{xcolor}
\usepackage{multirow,url}
\usepackage{alltt}
\usepackage{paralist}
\usepackage{threeparttable}
\usepackage{pifont}

\usepackage[ruled, linesnumbered]{algorithm2e}
\usepackage{tabulary}

\usepackage{tikz}

\newcommand{\indis}{indistinguishability}
\newcommand{\repset}{representative set}
\newcommand{\width}{diversity}

\newcommand{\set}[1]{\ensuremath{\left\{ #1 \right\}}}

\newtheorem{As}{Assumption}
\newtheorem{theorem}{Theorem}
\newtheorem{remark}{Remark}
\newtheorem{example}{Example}
\newtheorem{proposition}{Proposition}
\newtheorem{definition}{Definition}
\newtheorem{lemma}{Lemma}
\newtheorem{corollary}{Corollary}
\newtheorem{instance}{Instance}

\newcommand{\user}{\textsc{User}}
\newcommand{\task}{\textsc{Task}}
\newcommand{\enc}{\textsc{enc}}

\newcommand{\pat}{\textsc{pat}}

\newcommand{\vect}{\textsc{vec}}

\newcommand{\tick}{\ding{52}}

\def\showadditions{0}
\def\showdeletions{0}

\usepackage[normalem]{ulem}
\newcommand{\add}[1]{#1}
\if\showadditions1
\renewcommand{\add}[1]{{\bf{#1}}}
\fi

\newcommand{\remove}[1]{}
\if\showdeletions1
\renewcommand{\remove}[1]{{\color{red}\sout{#1}}}
\fi

\begin{document}

\title{Iterative Plan Construction for the Workflow Satisfiability Problem}

\author{\name David Cohen \email D.Cohen@rhul.ac.uk \\
\name Jason Crampton \email Jason.Crampton@rhul.ac.uk \\
\name Andrei Gagarin \email Andrei.Gagarin@rhul.ac.uk \\
\name Gregory Gutin \email G.Gutin@rhul.ac.uk \\
\name Mark Jones \email M.E.L.Jones@rhul.ac.uk \\
\addr Royal Holloway, University of London, UK }

\maketitle

\begin{abstract}
The \emph{Workflow Satisfiability Problem (WSP)} is a problem of practical interest that arises whenever tasks need to be performed by authorized users, subject to 
constraints defined by business rules.  We are required to decide whether there exists a \emph{plan} -- an assignment of tasks to authorized users -- such that all constraints are satisfied.
Several bespoke algorithms have been constructed for solving the WSP, optimised to deal with constraints (business rules) of particular types.

\remove{ The 
{\em Constraint Satisfaction Problem (CSP)} is a general paradigm for expressing, in a declarative format, problems where variables are to be assigned values from some domain.  The assignments are constrained by restricting the allowed simultaneous assignments to some sets of variables.  The CSP is in general NP-hard so there has been a considerable effort made into understanding the effect of restricting the type of allowed constraints.   This research program is nearing completion and we can strongly conjecture precisely which kinds of constraint language lead to polynomial solvability.}

It is natural to see the WSP as a subclass of the {\em Constraint Satisfaction Problem (CSP)} in which the variables are tasks and the domain is the set of users.
\remove{In this case we have a unary authorization constraint on every task.}
What makes the WSP distinctive as a CSP is that we can assume that the number of tasks is very small compared to the number of users.  This is in sharp contrast with traditional CSP models where the domain is small and the number of variables is very large.  
As such, it is appropriate to ask for which constraint languages the WSP is fixed-parameter tractable (FPT), parameterized by the number of tasks.

This novel approach to the WSP, using techniques for CSP,
has enabled us to generalise and unify existing algorithms. 
\add{Rather than considering algorithms for specific constraints, we 
design a generic algorithm which is a fixed-parameter algorithm
for several families of workflow constraints considered in the literature.}
We have identified a new FPT constraint language, user-independent constraint,   that includes many of the constraints of interest in business processing systems. We are also able to prove that the union of FPT languages remains FPT if they satisfy a simple compatibility condition.  
\remove{ This begins a theory of fixed parameter tractability of constraint languages for the WSP.}

In this paper we present our generic algorithm, in which plans are grouped into equivalence classes, each class being associated with a \emph{pattern}.
We demonstrate that our generic algorithm has running time
$O^*(2^{k\log k})$, where $k$ is the number of tasks, for the language of user-independent constraints. We also show 
\remove{that the generic algorithm is optimal for user-independent constraints in the following sense: } there is no algorithm of running time
$O^*(2^{o(k\log k)})$ for user-independent constraints unless the Exponential Time Hypothesis fails.

\end{abstract}

\section{Introduction}

\subsection{The Workflow Satisfiability Problem}

A business process is a collection of interrelated tasks that are performed by users in order to achieve some objective.
In many situations, we wish to restrict the users that can perform certain tasks.
In particular, we may wish to specify lists of users who are authorized to perform each of the workflow tasks.
Additionally, we may wish -- either because of the particular requirements of the business logic or security requirements -- to prevent certain combinations of users from performing particular combinations of tasks~\cite{Cr05}.
Such constraints include \emph{separation-of-duty} (also known as the ``two-man'' rule), which may be used to prevent sensitive combinations of tasks being performed by a single user, and \emph{binding-of-duty}, which requires that a particular combination of tasks is executed by the same user.
The use of constraints in workflow management systems to enforce security policies has been studied extensively in the last fifteen years; see~\cite{BeFeAt99,Cr05,WangLi10}, for example.

It is possible that the combination of constraints and authorization lists is ``unsatisfiable'', in the sense that there does not exist an assignment of users to tasks such that all contraints are satisfied and every task is performed by an authorized user; perhaps the minimal example being a requirement that two tasks are performed by the same user but the intersection of the authorization lists for these tasks is empty.
A plan that satisfies all constraints and allocates an authorized user to each task is said to be ``valid''.
The \emph{Workflow Satisfiability Problem (WSP)} takes a workflow specification as input and returns a valid plan if one exists and a null value otherwise.
It is important to determine whether a business process is satisfiable or not, since an unsatisfiable one can never be completed without violating the security policy encoded by the constraints and authorization lists.
Wang and Li~\cite{WangLi10} have shown, by a reduction from {\sc Graph Coloring}, that {\sc WSP} is an NP-hard subclass of CSP, even when we only consider binary separation-of-duty constraints.
So it is important that an algorithm that solves WSP is as efficient as possible~\cite[\S 2.2]{CrGu}.

\add{Many hard problems become less complex if some natural parameter of the instance is bounded.
Hence, we say a problem with input size $n$ and parameter $k$ is \emph{fixed-parameter tractable} (FPT) if it admits an algorithm with running time $O(f(k) n^d)$, where $d$ is a constant independent of $n$ and $k$, and $f$ is a computable function depending only on $k$.}\footnote{For an introduction to fixed-parameter algorithms and complexity, see, e.g., \cite{DowFel99,Nie06}.}
Wang and Li \cite{WangLi10} were the first to observe that fixed-parameter algorithmics is an appropriate way to study WSP,
because the number of tasks is usually small and often much smaller than the number of users.
It is appropriate therefore to consider fixed parameter algorithms for solving the WSP parameterized by the number of tasks, and to ask which constraint languages are fixed parameter tractable.
\remove{This paper begins a general theory of fixed parameter tractability of constraint languages for the WSP.  }

Wang and Li \cite{WangLi10}  proved that, in general, WSP is W[1]-hard and thus is highly unlikely to admit a fixed-parameter algorithm.
However, WSP is FPT if we consider only separation-of-duty and binding-of-duty constraints \cite{WangLi10}.
Crampton et al.~\cite{CrGuYeJournal} obtained significantly faster fixed-parameter algorithms that were applicable to ``regular'' constraints, thereby including the cases shown to be FPT by Wang and Li.
Subsequent research has demonstrated the existence of fixed-parameter algorithms for {\sc WSP} in the presence of other constraint types~\cite{CCGJR,CrGu}.
We define WSP formally and introduce a number of different constraint types, including regular constraints, in Section~\ref{sec:background}.

We will use the $O^*$ notation, which suppresses polynomial factors. That is, $g(n,k,m)  = O^*(h(n,k,m))$ if there exists a polynomial $q(n,k,m)$ such that $g(n,k,m) = O(q(n,k,m)h(n,k,m))$. In particular, an FPT algorithm is one that runs in time $O^*(f(k))$ for some computable function $f$ depending only on $k$.

\subsection{Relation Between WSP and CSP}

The \emph{Constraint Satisfaction Problem (CSP)} is a general paradigm for expressing, in a declarative format, problems where variables are to be assigned values from some domain.  The assignments are constrained by restricting the allowed simultaneous assignments to some sets of variables.  This model is useful in many application areas including planning, scheduling, frequency assignment and circuit verification~\cite{CPhandbook}.  The CSP community is a well-established research community dedicated to finding effective solution techniques for the CSP~\cite{DechterCP}.  

The CSP is NP-hard, even when only binary not-equals constraints are allowed and the domain has  three elements, as we can reduce {\sc Graph 3-Coloring} to the CSP.
\footnote{Wang and Li's NP-hardness result for WSP is thus a restatement of this well-known result for CSP.}
Hence, a considerable effort has been made to understand the effect of restricting the type of allowed constraints.   This research program is nearing completion and there is now strong evidence to support the algebraic dichotomy conjecture of Krokhin, Bulatov and Jeavons  characterising precisely which kinds of constraint language lead to polynomial solvability~\cite{KrBuJe}.

It is natural to see the WSP as a subclass of the CSP in which the variables are tasks and the domain is the set of users.
\add{In the language of CSP, we may view authorization lists as arbitrary unary constraints. Thus, WSP is a restriction of CSP in which the only constraint languages considered are those that allow all unary constraints.}
What makes the WSP distinctive as a CSP is that we can assume the number of tasks is very small compared to the number of users \cite{WangLi10}.  This is in sharp contrast with traditional CSP models where the domain is small and the number of variables is very large.  So it is natural to ask what the complexity of solving the WSP is when the number of tasks is a parameter $k$ of the problem,
\add{and we place restrictions on the allowed types of non-unary constraints.}

\subsection{Outline of the Paper}
Our novel approach to the WSP using techniques for CSP, characterising types of constraints as constraint languages with particular characteristics,  enables us to generalise and unify existing algorithms.
So, in this paper, for the first time, rather than considering algorithms for specific constraints, we 
design a generic algorithm which is a fixed-parameter algorithm
for several families of workflow constraints considered in the literature.
\add{In particular we introduce the notion of \emph{user-independent constraints}, which subsume a number of well-studied constraint types from the WSP literature, including the regular constraints studied by Crampton et al. \cite{CrGuYeJournal}.}

Our generic algorithm builds plans incrementally, discarding plans that can never satisfy the constraints.
It is based on a naive algorithm, presented in Section~\ref{sec:naive-solWSP}.
This naive algorithm stores 
far more information than is required to solve WSP,
so its running time is no better than exhaustively searching for a valid plan.

Our generic algorithm uses a general and classic paradigm: retain as little information as possible in every step of the algorithm.
This paradigm is used in such classical polynomial-time algorithms as Gaussian elimination for solving systems of linear equations and constraint propagation algorithms (used, for example, to solve 2SAT in polynomial time).
Our generic algorithm uses this paradigm in a problem-specific way, based on the concepts of extension-equivalence, plan-\indis{} and patterns, enabling us to retain a single pattern for each equivalence class of indistinguishable plans.
Extension-equivalence and plan encodings are described in Section~\ref{sec:ext-equiv}.


In Section~\ref{sec:solWSP}, we describe our pattern-based algorithm and demonstrate that it is
a fixed-parameter algorithm for WSP with user-independent constraints. We
show
the running time of our algorithm is $O^*(2^{k\log k})$ for WSP with
user-independent constraints and that there is no algorithm of running time
$O^*(2^{o(k\log k)})$ for WSP with user-independent constraints unless the Exponential Time Hypothesis\footnote{The Exponential Time Hypothesis claims there is no algorithm of running time $O^*(2^{o(n)})$ for 3SAT on $n$ variables \cite{ImPaZa01}.} (ETH) fails. Thus, unlike WSP with regular constraints (and problems studied in \cite{BoCyKrNe13,FoLoSa}), WSP with user-independent constraints is highly unlikely to admit an algorithm of
running time  $O^*(2^{O(k)}).$
To show that our generic algorithm is of interest for constraints other than user-independent, we prove that the generic algorithm is a single-exponential algorithm  for 
a constraint language obtained by an equivalence relation on the set of users.

In Section \ref{sec5} we show how our generic
algorithm can deal with unions of constraint languages.
This leads to a generalisation of our result for user-independent
constraints. 
We conclude
with
 discussions in Section \ref{sec6}.

\section{Background}\label{sec:background}


We define a \emph{workflow schema} to be a tuple $(S,U,A,C)$, where $S$ is the set of tasks in the workflow, $U$ is the set of users,  $A = \{A(s):\  s\in S\},$ where $A(s) \subseteq U$ is the \emph{authorization list} for task $s$, and $C$ is a set of workflow constraints. A \emph{workflow constraint} is a pair $c = (L,\Theta)$, where $L \subseteq S$ and $\Theta$ is a set of functions from $L$ to $U$: $L$ is the \emph{scope} of the constraint; $\Theta$ specifies those assignments of elements of $U$ to elements of $L$ that {\em satisfy} the constraint $c$.

Given $T \subseteq S$ and $X \subseteq U$, a \emph{plan} is a function $\pi: T \rightarrow X$.
Given a workflow constraint $(L,\Theta)$, $T \subseteq S$ and $X \subseteq U$, a plan $\pi: T \rightarrow X$ \emph{satisfies} $(L,\Theta)$ if
 either $L\setminus T\neq
\emptyset$ or $\pi|_L=\theta$ for some $\theta\in \Theta$.
 A plan $\pi: T \rightarrow X$ is \emph{eligible} if $\pi$ satisfies every constraint in $C$.
A plan $\pi: T \rightarrow X$ is \emph{authorized} if $\pi(s) \in A(s)$ for all $s \in T$.
 A plan is \emph{valid} if it is both authorized and eligible.
 A plan $\pi: S \rightarrow U$ is called a \emph{complete plan}.
 An algorithm to solve WSP takes a workflow schema $(S,U,A,C)$ as input and outputs a valid, complete plan, if one exists (and null, otherwise).


 \add{ As a running example, consider the following instance of WSP.}
 
 \begin{instance} \label{inst:mainExample}
 The task set $S = \{s_1,\dots, s_4\}$ and the user set $U = \{u_1, \dots, u_6\}$.
 The authorization lists are as follows (here a tick indicates that the given user is authorized for the given task):
 
 \begin{center}
    \begin{tabular}{|c | l | l | l | l | l | l |}
    \hline
    & $u_1$ & $u_2$ & $u_3$ & $u_4$ & $u_5$ & $u_6$ \\ \hline
    $s_1$ &\tick &\tick & & & & \\ \hline
    $s_2$ & & \tick& \tick& & & \\ \hline
    $s_3$ & & \tick& &\tick &\tick & \tick \\ \hline
    $s_4$ & & \tick& &\tick &\tick & \tick \\ \hline
    \end{tabular}
\end{center}

The constraints are as follows:
$s_1$ and $s_2$ must be assigned to the same user;
$s_2$ and $s_3$ must be assigned to different users;
$s_3$ and $s_4$ must be assigned to different users;
$s_1$ and $s_4$ must be assigned to different users.

 \end{instance}
 
 \add{
 We may denote the first constraint by $(s_1,s_2,=)$, and the last three constraints by $(s_2,s_3, \neq)$, $(s_3,s_4, \neq)$, $(s_1,s_4, \neq)$, respectively (see Subsection \ref{subsec:constraints}).
 
 \begin{example}
Consider Instance \ref{inst:mainExample}.

The following table gives the assignments for four plans $\pi_1, \pi_2, \pi_3, \pi_4$ (a dash indicates that the plan does not assign a task to any user):

 \begin{center}
    \begin{tabular}{|c | l | l | l | l | l | l | l |}
    \hline
    & $s_1$ & $s_2$ & $s_3$ & $s_4$ & Authorized & Eligible & Complete \\ \hline
    $\pi_1$ & $u_1$ &$u_2$ &$u_4$ &$u_5$ &\tick & &\tick \\ \hline
    $\pi_2$ & $u_1$&$u_1$ &$u_4$ &$u_5$ & &\tick &\tick \\ \hline
    $\pi_3$ & $u_1$& - &$u_4$ &$u_5$ &\tick &\tick & \\ \hline
    $\pi_4$ & $u_2$& $u_2$ &$u_4$ &$u_5$ &\tick &\tick &\tick \\ \hline
    \end{tabular}
\end{center}
  
  $\pi_1$ is a complete plan which is authorized but not eligible, as $s_1$ and $s_2$ are assigned to different users.
 
  $\pi_2$ is a complete plan which is eligible but not authorized, as $u_1$ is not authorized for $s_2$.
  
   $\pi_3$ is a plan which is authorized and eligible, and therefore valid.
   However, $\pi_3$ is not a complete plan as there is no assignment for $s_2$.
   
   $\pi_4$ is a complete plan which is eligible and authorized. Thus $\pi_4$ is a valid complete plan, and is therefore a solution to Instance \ref{inst:mainExample}.
 \end{example}
 
 }

 For an algorithm that runs on an instance $(S,U,A,C)$ of {\sc WSP}, we will measure the running time in terms of $n=|U|, k=|S|$, and $m=|C|$.
(The set $A$ of authorization lists consists of $k$ lists each of size at most $n$, so we do not need to consider the size of $A$ separately when measuring the running time.)
We will say an algorithm runs in polynomial time if it has running time at most $p(n,k,m)$, where  $p(n,k,m)$ is polynomial in $n,k$ and $m$.

\subsection{WSP Constraints}\label{subsec:constraints}
In this paper we are interested in the complexity of the WSP when the workflow constraint language (the set of permissible workflow constraints) is restricted.  In this section we introduce the constraint types of interest.  All of them have practical applications in the workflow problem.

We assume
that all constraints (including the unary authorization constraints) can be checked in polynomial time.  This means that it takes polynomial time to check whether any plan is authorized and whether it is valid.  The correctness of our algorithm is unaffected by this assumption, but choosing constraints not checkable in polynomial time would naturally affect the running time.
\remove{All the constraints types we define and use in this paper can be checked in polynomial time.}

%
%


\textbf{Constraints defined by a binary relation:} Constraints on two tasks, $s$ and $s'$, can be represented in the form $(s,s',\rho)$, where $\rho$ is a binary relation on $U$ \cite{Cr05}.
A plan $\pi$ satisfies such a constraint if \mbox{$\pi(s) \mathrel{\rho} \pi(s')$}.
Writing $=$ to denote the relation $\set{(u,u) : u \in U}$ and $\ne$ to denote the relation $\set{(u,v) : u,v \in U,u \ne v}$, separation-of-duty and binding-of-duty constraints may be represented in the form $(s,s',\ne)$ and $(s,s',=)$, respectively.
Crampton et al.~\cite{CrGuYeJournal} considered constraints for which $\rho$ is $\sim$ or $\nsim$, where $\sim$ is an \textit{equivalence relation} defined on $U$.
A practical example of such workflow constraints is when the equivalence relation partitions the users into different departments:   constraints could then enforce that two tasks be performed by members of the same department. Constraints that are not restricted to singleton tasks have also been considered~\cite{CrGuYeJournal,WangLi10}: 
a plan $\pi$ satisfies a constraint of the form $(S',S'',\rho)$ if there are tasks $s'\in S'$ and $s''\in S''$ such that \mbox{$\pi(s') \mathrel{\rho} \pi(s'')$}.

\textbf{Cardinality constraints:}
A {\em tasks-per-user counting constraint} has the form $(t_\ell,t_r,T)$, where $1 \leqslant t_\ell \leqslant t_r \leqslant k$ and $T\subseteq S$.  A plan $\pi$ satisfies $(t_\ell,t_r,T)$ if a user performs either no tasks in $T$ or between $t_\ell$ and $t_r$ tasks. Steps-per-user counting constraints generalize the cardinality constraints which have been widely adopted by the WSP community~\cite{ansi-rbac04,BeboFe01,JoBeLaGh05,SaCoFeYo96}. 

\textbf{Regular constraints:}
We say that $C$ is \emph{regular} if it satisfies the following condition: If a partition $S_1,\ldots ,S_p$ of $S$ is such that for every $i \in [p]$ there exists an eligible complete plan $\pi$ and user $u$ such that $\pi^{-1}(u) = S_i$,  then the plan $\bigcup_{i =1}^p (S_i \rightarrow u_i)$, where all $u_i​$'s are distinct, is eligible.
\add{Regular constraints extend
the set of constraints considered by Wang and Li \cite{WangLi10}.}
Crampton et al. \cite{CrGuYeJournal} show that the following constraints are regular: $(S',S'',\neq)$; $(S',S'',=)$, where at least one of the sets $S',S''$ is a singleton;  tasks-per-user counting constraints of the form $(t_\ell,t_r,T)$, where $t_\ell = 1$.

\textbf{User-Independent constraints:}
A constraint $(L,\Theta)$ is user-independent if whenever $\theta \in \Theta$  and $\psi : U \rightarrow U$ is a permutation then $\psi \circ \theta \in \Theta$.
In other words, user-independent constraints do not distinguish between users.

\add{ Many business rules are indifferent to the identities of the users that
complete a set of steps; they are only concerned with the relationships between
those users.  (Per-user authorization lists are the main exception to such rules.) 
The most obvious example of a user-independent constraint is the requirement that
two steps are performed by different users.  A more complex example might require
that at least three users are required to complete some sensitive set of steps.}

Every regular constraint is user-independent, but many user-independent constraints are not regular. Indeed, constraints of the type $(S',S'',=)$ are user-independent, but not necessarily regular \cite{CrGuYeJournal}. Many counting constraints in the Global Constraint Catalogue \cite{catalogue} are user-independent, but not regular. In particular, the constraint {\sc nvalue}, which bounds from above the number of users performing a set of tasks, is user-independent but not regular.
Note, however, that constraints of the form $(s',s'',\sim)$ and $(s',s'',\not\sim)$, are not user-independent, in general.

\add{
Authorization lists, when viewed as unary constraints, are not user-independent.
Thus for WSP with user-independent constraints, users are still distinguished due to the authorization lists.}

\subsection{A Naive Algorithm}\label{sec:naive-solWSP}

The main aim of this section is to present a simple algorithm (Algorithm \ref{alg:naive}) which will solve any instance of {\sc WSP}.
The running time of the algorithm is slightly worse than a brute-force algorithm, but the algorithm's basic structure provides a starting point from which to develop a more efficient algorithm.

Before proceeding further, we introduce some additional notation and terminology.

Let $\pi: T \rightarrow X$ be a  plan for some $T \subseteq S, X \subseteq U$.
Then let $\task(\pi) = T$ and $\user(\pi) = X$.
It is important for our generic algorithm that $\task(\pi)$ and $\user(\pi)$ are given as explicit parts of $\pi$. In particular, the set $\user(\pi)$ may be different from the set of users assigned to a task by $\pi$. That is, a user $u$ can be in $\user(\pi)$ without there being a task $s$ such that $\pi(s) = u$. It is worth observing that $\task(\pi)$ may be empty (because $\pi$ may not allocate any tasks to users in $X$).
 For any $T \subseteq S$ and $u \in U$, $(T \rightarrow u)$ denotes the  plan $\pi: T \rightarrow \{u\}$ such that $\pi(s)=u$ for all $s \in T$.
 
Two functions $f_1:\ D_1 \rightarrow E_1$ and $f_2:\ D_2\rightarrow E_2$ are \emph{disjoint} if $D_1\cap D_2=E_1\cap E_2=\emptyset$.
The {\em union} of two disjoint  functions $f_1: D_1\rightarrow E_1, f_2:D_2\rightarrow E_2$ is a function $f=f_1\cup f_2$ such that $f:D_1\cup D_2 \rightarrow E_1\cup E_2$ and $f(d)=f_{i}(d)$ for each $d\in D_{i},$ where $i\in \{1,2\}$. Let $g:\ D \rightarrow E$ and $h:\ E\rightarrow F$ be functions. Then $h\circ g$ denotes the composite function from $D$ to $F$ such that $h\circ g(d)=h(g(d))$ for each $d\in D.$ 
For an integer $p > 0$, the set $[p] = \{1,2, \dots, p\}$.

%
%
%
%
%
%
%
%
\remove{A pseudocode listing is presented in Algorithm~\ref{alg:naive}: given an ordering $u_1, \dots u_n$ of $U$, we construct, in order, sets $\Pi_i$ for each $i \in [n]$, where $\Pi_i$ is the set of all valid plans $\pi$ with $\user(\pi) = \{u_1, \dots, u_i\}$.
The correctness of the algorithm is easy to establish, as is its run-time (see the Appendix).}

\begin{algorithm}[!t]
 \SetKwFunction{Alg}{Win-Win}
\SetKwInOut{Input}{input}\SetKwInOut{Output}{output}

\Input{An instance $(S,U,A,C)$ of {\sc WSP}}

Construct an ordering $u_1, \dots, u_n$ of $U$\;

Set $\Pi_1 = \emptyset$\;
\ForEach{$T \subseteq S$}
{
  Set $\pi = (T \rightarrow u_1)$\;
  \If{$\pi$ is eligible and $u_1 \in A(s)$ for all $s \in T$}
  {
    Set $\Pi_1 = \Pi_1 \cup \{\pi\}$\;
  }
}
Set $i = 1$\;
\While{$i < n$}
{
  Set $\Pi_{i+1} = \emptyset$\;
  \ForEach{$\pi' \in \Pi_i$}
  {
    \ForEach{$T \subseteq S \setminus \task(\pi')$}
    {

	\If{$u_{i+1} \in A(s)$ for all $s \in T$}
	{
	  Set $\pi = \pi' \cup (T \rightarrow u_{i+1})$\;
	  \If{$\pi$ is eligible}
	  {
	    Set $\Pi_{i+1} = \Pi_{i+1} \cup \{\pi\}$\;
	  }
	}
    }
  }
  Set $i = i+1$\;
}
\ForEach{$\pi \in \Pi_n$}
{
  \If{$\task(\pi) = S$}
  {
    \Return $\pi$\;
  }
}
\Return {\sc null}\;
\caption{Naive solution for WSP}\label{alg:naive}
\end{algorithm}

\begin{proposition}\label{prop:naive}
 Let $(S,U,A,C)$ be an instance of {\sc WSP}, with $n = |U|$, $k = |S|$ and $m=|C|$.
Then $(S,U,A,C)$ can be solved in time $O^*((n+1)^k)$ by Algorithm \ref{alg:naive}.
\end{proposition}
\begin{proof}

 Let $u_1, \dots, u_n$ be an ordering of $U$, and let $U_i = \{u_1, \dots, u_i\}$ for each $i \in [n]$.
 For each $i \in [n]$ in turn, we will construct the set $\Pi_i$ of all  plans $\pi$ such that $\user(\pi) = U_i$ and $\pi$ is valid. If the set $\Pi_n$ contains no plan $\pi$ with $\task(\pi)=S$ then $(S,U,A,C)$ has no solution; otherwise, any such plan is a solution for $(S,U,A,C)$.
 
Algorithm \ref{alg:naive} shows how to construct the sets $\Pi_i$.
  It is not hard to verify that $\Pi_i$ contains exactly every valid plan $\pi$ with $\user(\pi) = U_i$, for each $i$. This implies the correctness of our algorithm. It remains to analyse the running time. 

 For each $i \in [n]$ and each $T \subseteq S$, there are at most $i^{|T|}$ valid  plans $\pi$ with $\user(\pi)=U_i, \task(\pi) = T$.  
 To construct $\Pi_1$, we need to consider all  plans $\pi$ with $\user(\pi)=U_1$, and there are exactly $2^k$ such  plans. For each  plan we can decide in polynomial time whether to add it to $\Pi_1$.
 To construct $\Pi_{i+1}$ for each $i \in [n-1]$, we need to consider every pair $(\pi', T)$ where $\pi' \in \Pi_i$ and  $T \subseteq S \setminus \task(\pi')$. 
 Consider the pair $(\pi',T)$, where $\pi'$ is an $(S',U_i)$-plan for some $S' \subseteq S$, and $T \subseteq S \setminus S'$.  Thus there are $i^{|S'|}$ possibilities for $\pi'$, and there are $2^{|S|-|S'|}$ choices for $T$.  Thus, the total number of pairs is given by $\sum_{S' \subseteq S} i^{|S'|} 2^{|S|-|S'|} = \sum_{j=0}^k \binom{k}{j} i^j 2^{k-j} = (i+2)^k$.
For each pair $(\pi', T)$ we can decide whether to add $\pi' \cup (T \rightarrow u_{i+1})$ to $\Pi_{i+1}$ in polynomial time. 
 Thus, to construct all $\Pi_i$ takes time  $O^*(\sum^{n-1}_{i=1}(i+2)^k) = O^*(n(n+1)^k) = O^*((n+1)^k)$.
\end{proof}

Algorithm \ref{alg:naive}  is inefficient even for small $k$,
\remove{In constructing $\Pi_{i+1}$,
 for each valid plan $\pi' \in \Pi_i$ we consider every possible extension of $\pi'$ to user $u_{i+1}$
In Algorithm \ref{alg:naive},} 
\add{due to the fact that} each $\Pi_i$ contains all valid plans $\pi'$ with $\user(\pi')=\{u_1,\dots,u_i\}$.
We show in the next section that it is not necessary to store so much information to solve the WSP.

\section {Plan-Indistinguishability Relations}\label{sec:ext-equiv}


We first introduce the notion of extension-equivalence, defined by an equivalence relation on the set of all plans.
Informally, the relation enables us to keep a single member of each equivalence class when building plans incrementally.

\begin{definition}\label{def:ee}
 Given an instance $(S,U,A,C)$ of {\sc WSP}, and two eligible plans $\pi_1$ and $\pi_2$, define $\pi_1 \approx \pi_2$ if the following conditions hold.
\begin{enumerate}
 \item $\user(\pi_1)=\user(\pi_2)$ and $\task(\pi_1)=\task(\pi_2)$;
 \item $\pi_1 \cup \pi'$ is eligible if and only if $\pi_2 \cup \pi'$ is eligible, for any plan $\pi'$ disjoint from $\pi_1$ and $\pi_2$.
\end{enumerate}
Then $\approx$ is an equivalence relation on the set of eligible plans and we say $\pi_1$ and $\pi_2$ are \emph{extension-equivalent} if $\pi_1 \approx \pi_2$.
\end{definition}

\add{
\begin{example}\label{ex:extEquiv}
 Consider Instance \ref{inst:mainExample}.
 
 Let $\pi_1: \{s_3,s_4\} \rightarrow \{u_2,u_4\}$ be the function such that $\pi_1(s_3) = u_2$ and $\pi_1(s_4)=u_4$.
 Let $\pi_2: \{s_3,s_4\} \rightarrow \{u_2,u_4\}$ be the function such that $\pi_2(s_3) = u_4$ and $\pi_2(s_4)=u_2$.
 
 Then $\pi_1$ and $\pi_2$ are both eligible, and $\user(\pi_1)=\user(\pi_2)$ and $\task(\pi_1)=\task(\pi_2)$. For any plan $\pi'$ disjoint from $\pi_1$ and $\pi_2$, the plan $\pi_1 \cup \pi'$ will satisfy the constraints  $(s_2,s_3, \neq), (s_1,s_4,\neq)$. Thus $\pi_1 \cup \pi'$ is eligible if and only if $\pi'$ is eligible. Similarly $\pi_2 \cup \pi'$ is eligible if and only if $\pi'$ is eligible. Thus $\pi_1 \cup \pi'$ is eligible if and only if $\pi_2 \cup \pi'$ is eligible, and so $\pi_1$ and $\pi_2$ are extension-equivalent. 
\end{example}
}

%

Suppose that we had a polynomial time algorithm to check whether two  eligible plans are extension-equivalent.
Then in Algorithm \ref{alg:naive}, we could keep track of just one  plan from each equivalence class: when constructing $\Pi_i$, we will only add $\pi_2$ to $\Pi_i$ if there is no $\pi_1$ extension-equivalent to $\pi_2$ already in $\Pi_i$; when we construct $\Pi_{i+1}$, we may use $\pi_1$ as a ``proxy'' for $\pi_2$.
%
If the number of extension-equivalent classes is small compared to the number of  plans, then the worst-case running time of the algorithm may be substantially lower than that of Algorithm~\ref{alg:naive}.

Unfortunately, it is not necessarily easy to decide if two eligible plans are extension-equivalent, so this approach is not practical.
However, we can always refine\footnote{An equivalence relation $\approx_2$ is a {\em refinement} of an equivalence relation $\approx_1$ if every equivalence class of $\approx_2$ is a subset of some equivalence class of $\approx_1.$} extension equivalence to an equivalence relation for which equivalence \textit{is} easy to determine.  For example, the identity equivalence relation where each plan is only equivalent to itself is such a refinement.

This refined equivalence relation may well have more equivalence classes than extension-equivalence, 
but substantially fewer than the identity relation, so we may
obtain a better running time than the naive algorithm.

\begin{definition}\label{def:similarity}
Given an instance $(S,U,A,C)$ of {\sc WSP},
 let $\Pi$ be the set of all eligible plans and let $\approx$ be an equivalence relation refining extension equivalence on $\Pi$.
 We say $\approx$ is a \emph{plan-\indis{} relation} (with respect to $C$) if, for all eligible $\pi_1, \pi_2$ such that $\pi_1 \approx \pi_2$, 
 and for any plan $\pi'$ disjoint from $\pi_1$ and $\pi_2$ such that $\pi_1\cup \pi'$ is eligible, we have that $\pi_1 \cup \pi' \approx \pi_2 \cup \pi'$.

 \end{definition}

%

 \remove{ As already observed, the identity relation is a trivial example of a plan-\indis{} relation.
  For regular constraints, the relation $\approx$ such that $\pi_1 \approx \pi_2$ if and only if $\user(\pi_1) = \user(\pi_2)$ and $\task(\pi_1) = \task(\pi_2)$ is a plan-\indis{} relation.
\vspace{1mm}
}
\add{
\begin{example}
Let $\approx$ be the identity relation on plans. That is, $\pi_1 \approx \pi_2$ if and only if $\user(\pi_1) = \user(\pi_2)$, $\task(\pi_1)=\task(\pi_2)$, and $\pi_1(s)=\pi_2(s)$ for all $s \in \user(\pi_1)$.
Then $\approx$ is a plan-\indis{} relation.
This shows that not every plan-\indis{} relation is the extension-equivalence relation. Indeed, the plans given in Example \ref{ex:extEquiv} are extension-equivalent but not identical.
 \end{example}
}
%
%
%
%
%

Recall that we refined extension equivalence since it may be hard to determine if two eligible plans are extension equivalent. It is therefore natural to assume the following:
 \begin{As}\label{ass:polycheckequiv}   Given a plan-\indis{} relation $\approx$, it takes polynomial time to check whether two eligible plans are equivalent under~$\approx$.
 \end{As}

The correctness of our algorithms does not depend on this assumption.  However, a poor choice of the plan-\indis{} relation could affect the running times.

We now describe appropriate plan-\indis{} relations for the constraints that we will be using.  In each case determining if two eligible plans are equivalent under~$\approx$ will take polynomial time.
%






\subsection{Plan-Indistinguishability Relation for User-Independent Constraints}

\begin{lemma}\label{lem:uirelation}
Suppose all constraints are user-independent, and
let $\approx_{ui}$ be a relation such that $\pi_1 \approx_{ui} \pi_2$ if and only if
 \begin{enumerate}
  \item $\user(\pi_1)=\user(\pi_2)$ and $\task(\pi_1)=\task(\pi_2)$;
  \item For all $s,t \in \task(\pi_1)$, $\pi_1(s)=\pi_1(t)$ if and only if $\pi_2(s)=\pi_2(t)$.
 \end{enumerate}
Then $\approx_{ui}$ is a plan-\indis{} relation on the set of eligible plans.
\end{lemma}

\begin{proof}

By definition of user-independent constraints, if $\pi$ is an eligible plan and  $\psi : U \rightarrow U$ is a permutation, then $\psi\circ\pi$ is also eligible.
Suppose that  $\pi_1 \approx_{ui} \pi_2$ and let $T=\task(\pi_1)$ and $X=\user(\pi_1)$. Let $\psi':\pi_1(T) \rightarrow \pi_2(T)$ be a function such that  $\psi'(\pi_1(t)) = \pi_2(t)$ for any task~$t$. 
Let $\psi'': X\setminus \pi_1(T)
\rightarrow X\setminus \pi_2(T) $ be an arbitrary bijection (note that
$|\pi_1(T)|=|\pi_2(T)|$ by Condition 2 of $\approx_{ui}$).
Let $\psi = \psi' \cup \psi''$.
Then $\psi$ is a permutation such that $\pi_2 = \psi \circ \pi_1$. Thus $\pi_1$ is eligible if and only if $\pi_2$ is eligible.

 Now consider two eligible plans $\pi_1, \pi_2$ such that $\pi_1 \approx_{ui} \pi_2$, and a plan $\pi'$ disjoint from $\pi_1$ and $\pi_2$.
First we show that $\pi_1 \cup \pi' \approx_{ui} \pi_2 \cup \pi'$. It is clear that $\user(\pi_1 \cup \pi') = \user(\pi_2 \cup \pi')$ and $\task(\pi_1 \cup \pi') = \task(\pi_2 \cup \pi')$. Now for any $s,t \in \user(\pi_1 \cup \pi')$, if $(\pi_1 \cup \pi')(s) = (\pi_1 \cup \pi')(t)$, then either $s,t$ are both in $\task(\pi')$, in which case $(\pi_2 \cup \pi')(s) = (\pi_2 \cup \pi')(t)$ trivially, or $s,t$ are both in $\task(\pi_1)$, in which case $\pi_2(s)=\pi_2(t)$ since $\pi_1 \approx_{ui} \pi_2$, and hence $(\pi_2 \cup \pi')(s) = (\pi_2 \cup \pi')(t)$. Thus if  $(\pi_1 \cup \pi')(s) = (\pi_1 \cup \pi')(t)$ then $(\pi_2 \cup \pi')(s) = (\pi_2 \cup \pi')(t)$, and by a similar argument the converse holds. Thus $\pi_1 \cup \pi' \approx_{ui} \pi_2 \cup \pi'$.
 Furthermore, it follows by the argument in the first paragraph that  $\pi_1 \cup \pi'$ is eligible if and only if $\pi_2 \cup \pi'$ is eligible. Thus, the condition of Definition \ref{def:similarity} and the second condition of Definition \ref{def:ee} hold.

The first condition of $\approx_{ui}$ trivially satisfies the first condition of Definition \ref{def:ee}. 
 Thus, $\approx_{ui}$ satisfies all the conditions of a plan-\indis{} relation.
\end{proof}

\add{
\begin{example}\label{ex:uiExample}
 Consider an instance of {\sc WSP} with users $u_1, \dots u_6$ and tasks $s_1, \dots, s_6$ in which all constraints are user-independent. Let $\approx_{ui}$ be the plan-\indis{} relation given by Lemma \ref{lem:uirelation}.
 Let $c_1$ be the constraint with scope $\{s_2,s_3,s_4,s_5\}$ such that $c_1$ is satisfied if and only if an even number of users are assigned to tasks in $\{s_2,s_3,s_4,s_5\}$.
 Let $c_2$ be the constraint with scope $\{s_1,s_3,s_4,s_6\}$ such that $c_2$ is satisfied if and only if either $s_1$ and $s_3$ are assigned to different users, or $s_4$ and $s_6$ are assigned to different users.
 Suppose that $c_1$ and $c_2$ are the only constraints whose scope contains tasks from both $\{s_1,s_2,s_3\}$ and $\{s_4,s_5,s_6\}$.
 
 Now consider the plans $\pi_1,\pi_2: \{s_1,s_2,s_3\} \rightarrow \{u_1,u_2,u_3,u_4\}$ such that $\pi_1(s_1)=u_1, \pi_1(s_2)=u_2, \pi_1(s_3)=u_1$, and $\pi_2(s_1)=u_3, \pi_2(s_2)=u_4, \pi_2(s_3)=u_3$, and suppose that $\pi_1,\pi_2$ are both eligible.
 Then $\pi_1$ and $\pi_2$ are equivalent under $\approx_{ui}$.
 
 Observe that for any plan $\pi'$ disjoint from $\pi_1$ and $\pi_2$, $\pi_1 \cup \pi'$ is eligible if and only if $\pi_2 \cup \pi'$ is eligible. As $\pi_1$ and $\pi_2$ both assign two users to $\{s_2,s_3\}$, $\pi'$ must assign two users to $\{s_4,s_5\}$ in order to satisfy $c_1$. As $\pi_1$ and $\pi_2$ both assign $s_1$ and $s_3$ to the same user, $\pi'$ must assign $s_4$ and $s_5$ to different users in order to satisfy $c_2$. 
 As long as these conditions are satisfied, and $\pi'$ satisfies all constraints with scope in $\{s_4,s_5,s_6\}$, then $\pi_1 \cup \pi'$ and $\pi_2 \cup \pi'$ will both be eligible.
\end{example}
}

\subsection{Plan-Indistinguishability Relation for Equivalence Relation Constraints}

Recall that given a binary relation $\rho$ on $U$, a constraint of the form $(s_i,s_j, \rho)$ is satisfied by a plan $\pi$ if $\pi(s_i)~\rho~\pi(s_j)$. 
Recall that such constraints are not user-independent in general.

\begin{lemma}\label{lem:equivRelation} 
 Suppose $\sim$ is an equivalence relation on $U$.
 Let $V_1, \dots, V_l$ be the equivalence classes of $\sim$ over $U$.
Suppose all constraints are of the form $(s_i,s_j, \sim)$ or $(s_i,s_j, \not\sim)$.
Let $\approx_{e}$ be a relation such that $\pi_1 \approx_{e} \pi_2$ if and only if
 \begin{enumerate}
  \item $\user(\pi_1)=\user(\pi_2)$ and $\task(\pi_1)=\task(\pi_2)$;
  \item For all equivalence classes $V_j$
  such that $V_j \cap \user(\pi_1) \neq \emptyset$ and $V_j \setminus \user(\pi_1) \neq \emptyset$,
  we have that
  for all $s \in \task(\pi_1)$, $\pi_1(s) \in V_j$ if and only if $\pi_2(s) \in V_j$. 
 \end{enumerate}
Then $\approx_{e}$ is a plan-\indis{} relation.

\end{lemma}

\begin{proof}
 It is clear that $\approx_{e}$ satisfies the first condition of Definition \ref{def:ee}.
 Now suppose $\pi_1, \pi_2$ are eligible plans such that $\pi_1 \approx_{e} \pi_2$, and let $\pi'$ be a plan disjoint from $\pi_1$ and $\pi_2$. We first show that $\pi_1 \cup \pi'$ is eligible if and only if $\pi_2 \cup \pi'$ is eligible.

 Suppose that $\pi_1 \cup \pi'$ is eligible.
Consider two tasks $t, t' \in \task(\pi_2 \cup \pi')$.  If $\set{t,t'} \subseteq \task(\pi')$ then $\pi_2 \cup \pi'$ will not falsify any constraint on $t$ and $t'$ since it is equal to $\pi_1 \cup \pi'$ when restricted to $\{t,t'\}$ and $\pi_1 \cup \pi'$ is eligible.  If $\set{t,t'} \subseteq \task(\pi_2)$, then $\pi_2 \cup \pi'$ will not break any constraints since $\pi_2$ is eligible.

So we may assume that $t \in \task(\pi_2)$, $t' \in \task(\pi')$.  By definition, $(\pi_2 \cup \pi')(t) \sim (\pi_2 \cup \pi')(t')$ if and only if there exists $j \in [l]$ such that $\pi_2(t), \pi'(t') \in V_j$. But then $V_j \cap \user(\pi_2) \neq \emptyset$ and $V_j \setminus \user(\pi_2) \neq \emptyset$. Therefore, by definition of $\approx_{e}$, $\pi_1(s) \in V_j$ if and only if $\pi_2(s) \in V_j$, for all $s \in \task(\pi_1)$. In particular, $\pi_1(t) \in V_j$, and so $(\pi_1 \cup \pi')(t) \sim (\pi_1 \cup \pi')(t')$. By a similar argument, if $(\pi_1 \cup \pi')(t) \sim (\pi_1 \cup \pi')(t')$ then $(\pi_2 \cup \pi')(t) \sim (\pi_2 \cup \pi')(t')$. Therefore, every constraint is satisfied by $(\pi_1 \cup \pi')$ if and only if it is satisfied by $(\pi_2 \cup \pi')$.
Therefore if $\pi_1 \cup \pi'$ is eligible then so is $\pi_2 \cup \pi'$, and by a similar argument the converse holds.

It remains to show that $\pi_1 \cup \pi' \approx_{e} \pi_2 \cup \pi'$.
It is clear that the user and task sets are the same.
As they have the same user set, the sets
$\set{V_j : V_j \cap \user(\pi_1 \cup \pi') \neq \emptyset, V_j \setminus \user(\pi_1 \cup \pi') \neq \emptyset}$ and $\set{V_j : V_j \cap \user(\pi_2 \cup \pi') \neq \emptyset, V_j \setminus \user(\pi_2 \cup \pi') \neq \emptyset}$ are the same.
Furthermore, for each $V_j$ in this set and any $s \in \task(\pi_1 \cup \pi')$, 
if $(\pi_1 \cup \pi')(s) \in V_j$ then $(\pi_2 \cup \pi')(s) \in V_j$, as either $s \in \task(\pi_1)$, in which case $V_j \cap \user(\pi_1) \neq \emptyset, V_j \setminus \user(\pi_1)  \neq \emptyset$ and so $\pi_2(s) \in V_j$, or $s \in \task(\pi')$, in which case  $(\pi_2 \cup \pi')(s) = \pi'(s) = (\pi_1 \cup \pi')(s)$. By a similar argument, if $(\pi_2 \cup \pi')(s) \in V_j$ then $(\pi_1 \cup \pi')(s) \in V_j$. Thus $\pi_1 \cup \pi' \approx_{e} \pi_2 \cup \pi'$.
%
\end{proof}

\add{

\begin{example}\label{ex:equivExample}
 Let $\sim$ be an equivalence relation on users with equivalence classes $\{u_1\}, \{u_2\}, \{u_3,u_4, u_5\}, \{u_6, u_7,u_8\}$.
 Consider an instance of {\sc WSP} with users $u_1, \dots, u_8$ and tasks $s_1, \dots, s_6$ in which all constraints are of the form $(s_i,s_j, \sim)$ or $(s_i,s_j, \not\sim)$. Let $\approx_{e}$ be the plan-\indis{} relation given by Lemma \ref{lem:equivRelation}.
 Suppose that the only constraints whose scope contains tasks from both $\{s_1,s_2,s_3\}$ and $\{s_4,s_5,s_6\}$ are the constraints $(s_1,s_5, \not\sim)$,  $(s_2,s_5, \sim)$ and $(s_2,s_6,\not\sim)$.
 
  Now consider the plans $\pi_1,\pi_2: \{s_1,s_2,s_3\} \rightarrow \{u_1,u_2,u_3,u_4\}$ such that $\pi_1(s_1)=u_1, \pi_1(s_2)=u_3, \pi_1(s_3)=u_3$, and $\pi_2(s_1)=u_2, \pi_2(s_2)=u_3, \pi_2(s_3)=u_4$, and suppose that $\pi_1,\pi_2$ are both eligible.
 Then $\pi_1$ and $\pi_2$ are equivalent under $\approx_{e}$.
 
  Observe that for any plan $\pi'$ disjoint from $\pi_1$ and $\pi_2$, $\pi_1 \cup \pi'$ is eligible if and only if $\pi_2 \cup \pi'$ is eligible. 
  The only $\sim$-equivalence class with members in $\{u_1,u_2,u_3,u_4\}$ and members not in $\{u_1,u_2,u_3,u_4\}$ is the class $\{u_3,u_4,u_5\}$.
  $\pi_1$ and $\pi_2$ both assign members of $\{u_3,u_4\}$ to exactly the set $\{s_2,s_3\}$.
  Thus for any plan $\pi'$ disjoint from $\pi_1$ and $\pi_2$, $\pi_1 \cup \pi'$ and $\pi_2 \cup \pi'$ will both satisfy the constraint $(s_1,s_5,\not\sim)$ whatever $\pi'$ assigns to $s_5$. 
  They will both satisfy  $(s_2,s_5, \sim)$ only if $\pi'$ assigns $s_5$ to $u_5$, and they will both satisfy $(s_2,s_6, \not\sim)$ only if $\pi'$ does not assign $s_6$ to $u_5$.
 As long as these conditions are satisfied, and $\pi'$ satisfies all constraints with scope in $\{s_4,s_5,s_6\}$, then $\pi_1 \cup \pi'$ and $\pi_2 \cup \pi'$ will both be eligible.
\end{example}

}

 \section{A Generic Algorithm for {\sc WSP}}\label{sec:solWSP}

In what follows, for each $X \subseteq U, T \subseteq S$, we let $\Pi[X,T]$ denote the set of  all eligible plans $\pi$ with $ \user(\pi)=X$ an $\task(\pi)=T$.
 In this section we will introduce an algorithm that works in a similar way to Algorithm \ref{alg:naive}, except that instead of storing all valid plans over a particular set of users or tasks, 
 we will construct $\Pi[X,T]$-\repset s for each task set $T$ and certain user sets $X$.
\add{By definition, the
equivalence classes of any plan-indistinguishability relation necessarily partition
$\Pi[X,T]$.  Hence any such equivalence class has a representation of the form
$(X,T,\cdot)$, where $\cdot$ is dependent on the constraint language.  In the remainder
of this section we describe the algorithm and give examples of these
representations.}

\subsection{Encodings and patterns}

\add{In our generic algorithm, we will construct plans iteratively, using at most one plan from each equivalence class under a plan-\indis{} relation. The running time of the algorithm will depend on the number of equivalence classes under this relation, over certain sets of plans. 
To ensure that sets of equivalence classes can be ordered and therefore searched and sorted efficiently,
we introduce the notion of encodings and patterns. Loosely speaking, an \emph{encoding} is a function that maps all the plans in a $\approx$-equivalence class to the same element (the \emph{pattern} of those plans). These encodings ensure logarithmic-time access and insertion operations into a representative set of plans, rather than the linear time that a naive method would allow.

Note that the use of encodings and patterns is not necessary for any of our fixed-parameter tractability results; the same problems could be solved without the use of patterns and encodings in fixed-parameter time, but the function in $k$ would grow more quickly.}

 
%
%


\remove{
The time taken for searching a set for an element or inserting an element into a set is polynomial in the size of a 
set.  Therefore, when considering the running-time of fixed-parameter 
algorithms and using $O^*$ notation, we may ignore the costs of searching in a set and inserting into a set if its cardinality is polynomial in the size of the inputs to the 
problem.  In the case of WSP, the size of the sets of users, tasks and 
constraints are inputs to the problem, but the number of plans is 
$O(n^k)$ in the worst case.  Thus, searching in a set of plans and inserting into a set of plans (or a 
representative set) must be explicitly accounted for in the running time 
of an algorithm to solve WSP.}

%

\remove{For example, we will often want to determine, given a plan $\pi$ and some set of plans $\Pi^*$, whether $\Pi^*$ contains a plan $\pi'$ with $\pi' \approx \pi$.
This can naively be done in time $O^*(|\Pi^*|)$ by checking whether $\pi' \approx \pi$ for every $\pi' \in \Pi^*$.
However, using the techniques described below, it is possible to do this in time  $O^*(\log(|\Pi^*|))$.
This makes a big difference when the size of $\Pi^*$ is, for example, $2^{O(k \log k)}$, as will be the case when we consider user-independent constraints.}

\remove{Thus to improve the running time of our algorithm, we introduce the notion of encodings and patterns. Loosely speaking, an \emph{encoding} is a function that maps all the plans in a $\approx$-equivalence class to the same element (the \emph{pattern} of those plans).
By identifying all $\approx$-equivalent plans with the same pattern, we allow for some savings in the running time of our set operations. In particular, checking whether a set of plans contains a plan equivalent to a particular plan $\pi$ can be done more efficiently using patterns.
}

%


\begin{definition}\label{def:pat}
Given an instance $(S,U,A,C)$ of {\sc WSP} and a plan-\indis{} relation $\approx$,
let $\Pi$ be the set of all  plans.
 Let $\pat$ be some set and consider a function $\enc: \Pi \rightarrow \pat$.
 For any $X \subseteq U, T \subseteq S$, let $\pat[X,T] = \enc(\Pi[X,T])$.
 Then we say $\enc$ is a \emph{$\approx$-encoding} (or an \emph{encoding for $\approx$}) if, for  any $X \subseteq U, T \subseteq S$ and any $\pi_1, \pi_2 \in \Pi[X,T]$, we have that
 \begin{enumerate}
  \item $\enc(\pi_1)=\enc(\pi_2)$ if and only if $\pi_1 \approx \pi_2$;
  \item $\enc(\pi_1)$ can be calculated in time polynomial in $n,k,m$;
  \item There exists a linear ordering $\preceq$ on $\pat[X,T]$ such that, for $p,p' \in \pat[X,T]$, we can decide whether $p \preceq p'$ in time polynomial in $n,k,m$.
 \end{enumerate}
The elements of $\pat$ are called \emph{$\approx$-patterns}. If $\enc(\pi)=p$ then we say $p$ is the \emph{$\approx$-pattern of $\pi$}.
\end{definition}

The second and third conditions of Definition~\ref{def:pat} ensure that we may use encodings to organise our plans in  a reasonable time.
When $\approx$ is clear from the context, we will refer to a $\approx$-encoding as an \emph{encoding} and $\approx$-patterns as \emph{patterns}.

We note some  complexity consequences of Definition~\ref{def:pat} in the following:

\begin{proposition}\label{ass:settimes}
For an encoding of a plan-\indis{} relation $\approx$ \add{and a set of patterns $\pat^*$}, by assigning patterns in $\pat^*$ to the nodes of a balanced binary tree, we can perform the following two operations in time $O^*(\log(|\pat^*|))$: 
(i) check whether $p \in \pat^*$, and (ii) insert a pattern $p\notin \pat^*$ into $\pat^*$ .
\end{proposition}
\begin{proof}
Recall that comparisons are polynomial in $n,k, m$. Now our result follows from the well-known properties  of balanced binary trees, see, e.g., \cite{CLRS01}. 
\end{proof}

%
%
%
%
%

We now show that the plan-\indis{} relations given in the previous section have encodings.  We first need to define a lexicographic ordering.

\begin{definition}
Given a totally ordered set $(A,\leq)$, the (total) 
\emph{lexicographic order} $\leq$ on $d$-tuples from $A^d$ is defined as follows.
We say that $(x_1,\dots,x_d) \leq (y_1,\dots, y_d)$ if 
either $x_j=y_j$ for all $j\in [d]$ or there is an $i$ with $x_i<y_i$ such
that $x_j=y_j$ for all $j<i$.

Taking $A = \mathbb{N}$ and $d = k$ we obtain the natural lexicographic order on $\mathbb{N}_0^k$.
\end{definition}

We can also lexicographically order the sets of disjoint subsets of an ordered set $T = \{t_1,\dots,t_k\}$, where $t_1 < \dots < t_k$.

\begin{definition}
We associate a $k$-tuple  $(x_1,\dots , x_k) \in \mathbb{N}_0^k$ with each set ${\cal S}$ of disjoint subsets $\{S_1,\ldots,S_r\}$ of $\{t_1,\dots,t_k\}$
as follows. We have $x_i =0$ if $t_i \notin \cup_{m=1}^{r} S_m$. For $t_i \in \cup_{m=1}^{r} S_m$,
\begin{itemize}
\item if there are $j<i$ and $m$ such that $\{t_i, t_j\} \subseteq S_m$ then $x_i = x_j$,
\item otherwise $x_{i} = \max\{x_1,\dots,x_{i-1}\} + 1$, where $\max \emptyset = 0$. 
\end{itemize}
We will write $\vect({\cal S})=(x_1,\dots , x_k)$. 
Note that $\vect({\cal S})$ can be computed  in time $O(k^2)$.
\end{definition}

For example, for $T = \{1,\dots,8\}$ and the sets ${\cal A} = \{\{2,4\},\{3\},\{5,7\}\}$ and ${\cal B} = \{\{2,3,4\},\{5,7\}\}$, we have $\vect({\cal A})=(0,1,2,1,3,0,3,0)$ and  $\vect({\cal B})=(0,1,1,1,2,0,2,0)$.  So $\cal A$ is lexicographically bigger than $\cal B$.

\begin{corollary}\label{cor:uiencoding}
 Let $\approx_{ui}$ be the plan-indistinguishability relation given for a set of user-independent constraints in Lemma \ref{lem:uirelation}. Then there exists an encoding for $\approx_{ui}$.
\end{corollary}
\begin{proof}
 Let $s_1,\ldots ,s_k$ be an ordering of $S$ and $\pi$ a plan. Let ${\cal S}^{\pi}=\{\pi^{-1}(u):\ u\in \user(\pi)\}$  and let $\vect(\pi)=\vect({\cal S}^{\pi}).$
For a  plan $\pi$, 
let $\enc(\pi)$ be the tuple $(\user(\pi), \task(\pi), \vect(\pi))$. 

It is clear that $\enc(\pi_1) = \enc(\pi_2)$ if and only if $\pi_1 \approx_{ui} \pi_2$, as $\pi_r(s_i) = \pi_r(s_j)$ if and only if $y_i = y_j$ in $\vect(\pi_r) = (y_1,\dots,y_k)$, for $r \in \{1,2\}$. Furthermore it is clear that $\enc(\pi)$ can be determined in polynomial time for any $\pi$. 

It remains to define a linear ordering on $\pat[X,T]$ for a given $X \subseteq U, T \subseteq S$.
For two patterns $p =  (X, T, (x_1, \dots, x_k)), p' =  (X, T, (y_1, \dots, y_k)) \in \pat[X,T]$, we define $p \preceq p'$ if $(x_1, \dots, x_k) \leq (y_1, \dots, y_k)$.
\end{proof}

\begin{example}
 Let $\enc$ be the encoding given in the proof of Corollary \ref{cor:uiencoding}.
 Let $\pi_1,\pi_2$ be the plans given in Example \ref{ex:uiExample}.
 Then $\enc(\pi_1) = \enc(\pi_2) =\{\{u_1,u_2,u_3, u_4\}, \{s_1,s_2,s_3\}, (1,2,1,0,0,0)\}$.
\end{example}

\begin{corollary}\label{cor:nestedEquivEncodingencoding}
 Let $\approx_{e}$ be the plan-indistinguishability relation given for a set of constraints on equivalence relations in Lemma \ref{lem:equivRelation}. Then there exists an encoding for $\approx_{e}$.
\end{corollary}
\begin{proof}
Suppose $\sim$ is an equivalence relation on users, and let $V_1, \dots, V_p$ be the equivalence classes of $\sim$ over $U$.
Suppose all constraints are of the form $(s_i,s_j, \sim)$ or $(s_i,s_j, \not\sim)$.

 For a plan $\pi$, define $\enc(\pi)$ to be
$(\user(\pi), \task(\pi), {\cal T}^{\pi}),$ where $$ {\cal T}^{\pi}=\set{\pi^{-1}(V_j\cap \user(\pi)) :\ V_j \cap \user(\pi) \neq \emptyset, V_j \setminus \user(\pi) \neq \emptyset}.$$

It is clear that $\enc(\pi_1) = \enc(\pi_2)$ if and only if $\pi_1 \approx_{e} \pi_2$, as $\pi_i(s) \in V_j$ if and only if $s \in \pi_i^{-1}(V_j)$, for $i \in \{1,2\}$. Furthermore it is clear that $\enc(\pi)$ can be determined in polynomial time for any $\pi$. 

It remains to define a linear ordering on $\pat[X,T]$ for a given $X \subseteq U, T \subseteq S$.
Let $\pi:\ T\rightarrow X$ be a plan.
As  ${\cal T}^{\pi}$ is a set of disjoint subsets
of 
$\task(\pi)$, and $T$ has a natural order,  we can order patterns in $\pat[X,T]$ according to the lexicographic order of ${\cal T}^{\pi}$.

\end{proof}

\begin{example}
 Let $\enc$ be the encoding given in the proof of Corollary \ref{cor:nestedEquivEncodingencoding}.
 Let $\pi_1,\pi_2$ be the plans given in Example \ref{ex:equivExample}.
 Then $\enc(\pi_1) = \enc(\pi_2) =\{\{u_1,u_2,u_3, u_4\}, \{s_1,s_2,s_3\}, \{\{s_2,s_3\}\}\}$.
\end{example}

\subsection{The Generic Algorithm}

We use the notion of \emph{\width{} } introduced in the next
definition to analyse the running time of our generic algorithm.

\begin{definition}\label{def:width}
 Let $(S,U,A,C)$ be an instance of {\sc WSP}, with $n=|U|, k=|S|$ and $m=|C|$, and suppose $\approx$ is a plan-\indis{} relation with respect to $C$.
 Given an ordering $u_1, \dots, u_n$ of $U$, let $U_i = \{u_1, \dots, u_i\}$ for each $i \in [n]$. 
 Let $w_i$ be the number of equivalence classes of $\approx$ over  the set $\Pi[U_i, T]$ of eligible plans.
 Then we define the \emph{\width{}  of $\approx$ with respect to $u_1, \dots, u_n$} to be $w = \max_{i \in [n]} w_i$.
\end{definition}

Since our generic algorithm 
only stores one plan from each equivalence class under $\approx$, 
we need the notion of a \repset{}.
 
\begin{definition}
Given an instance $(S,U,A,C)$ of {\sc WSP},
  let $\Pi'$ be a set of eligible plans and let $\approx$ be a plan-\indis{} relation. A set $\Pi''$ is said to be a \emph{$\Pi'$-\repset{} with respect to $\approx$} if the following properties hold:
  \begin{enumerate}
   \item $\Pi'' \subseteq \Pi'$; every plan in $\Pi''$ is valid;
   \item for every valid $\pi' \in \Pi'$, there exists a $\pi'' \in \Pi''$ such that $\pi' \approx \pi''$.
  \end{enumerate}

 \end{definition}

When $\approx$ is clear from context, we will say $\Pi''$ is a \emph{$\Pi'$-\repset{}} or a \emph{\repset{} for $\Pi'$}.
Our generic algorithm is based on finding plan-\indis{} relations for which there exist small \repset s.

\begin{theorem}\label{thm:linear}
 Let $(S,U,A,C)$ be an instance of {\sc WSP}, with $n=|U|, k=|S|$ and $m=|C|$.
 Let $u_1, \dots, u_n$ be an ordering of $U$, and let $U_i = \{u_1, \dots, u_i\}$ for each $i \in [n]$, and $U_0 = \emptyset$.
Suppose $\approx$ has \width{} $w$ with respect to $u_1, \dots, u_n$.
 Furthermore suppose that there exists a $\approx$-encoding $\enc$.
Then $(S,U,A,C)$ can be solved in time 
$O^*(3^kw\log w)$.
\end{theorem}


\begin{algorithm}[!t]
 \SetKwFunction{Alg}{Win-Win}
\SetKwInOut{Input}{input}\SetKwInOut{Output}{output}

\Input{An instance $(S,U,A,C)$ of {\sc WSP}, an ordering $u_1, \dots, u_n$ of $U$, a plan-\indis{} relation $\approx$}

Set $\Pi[U_0, \emptyset]^* = \{(\emptyset \rightarrow \emptyset)\}$\;
\ForEach{$\emptyset \neq T \subseteq S$}
{
  Set $\Pi[U_0, T]^* = \emptyset$\;
}

Set $i = 0$\;
\While{$i < n$}
{
  \ForEach{$T \subseteq S$}
  {
     Set $\Pi[U_{i+1}, T]^* = \emptyset$\;  
     Set $\pat[U_{i+1},T]^* = \emptyset$\;
     \ForEach{$T' \subseteq T$}
     {
      Set $T'' = T \setminus T'$\;
      \If{$u_{i+1}\in A(s)$ for all $s \in T''$}
      {  
	\ForEach{$\pi' \in \Pi[U_i, T']^*$}
	{
	  Set $\pi = \pi' \cup (T'' \rightarrow u_{i+1})$\;
	  \If{$\pi$ is eligible}
	  {
	    Set $p = \enc(\pi)$\;
	    \If{$p \notin \pat[U_{i+1},T]^*$}
	    {
	      Insert $p$ into $\pat[U_{i+1},T]^*$\;
	      Set $\Pi[U_{i+1},T]^*= \Pi[U_{i+1},T]^* \cup \{\pi\}$\;
	    }
	  }
	}
      }
    }  
  }
  Set $i = i+1$\;
}
\eIf{$\Pi[U_n,S]^* \neq \emptyset$}
{
  \Return $\pi \in \Pi[U_n,S]^*$\;
}
{
  \Return {\sc null}\;
}
\caption{Generic algorithm for WSP}\label{alg:generic1}
\end{algorithm}


\begin{proof}

The proof proceeds by proving correctness and
bounding the running time of Algorithm \ref{alg:generic1}, which solves WSP.  
To begin the proof, we give an overview of Algorithm \ref{alg:generic1}.

\begin{itemize}
 \item  For each $i \in [n]$ in turn and each $T \subseteq S$, we will construct a \repset{} for $\Pi[U_i,T]$, denoted by $\Pi[U_i,T]^*$.
 \item As well as constructing the set $\Pi[U_i,T]^*$, we also maintain a companion set $\pat[U_i,T]^*$ = $\enc(\Pi[U_i,T]^*)$. 
 This provides an efficient way of representing the equivalence classes of $\Pi[U_i,T]^*$. In particular, it
allows us to check whether a given valid plan $\pi$ should be added to $\Pi[U_i,T]^*$, 
faster than by searching $\Pi[U_i,T]^*$ linearly. 
 \item After $\Pi[U_n,S]^*$ has been constructed, it remains to check whether $\Pi[U_n,S]^*$ is non-empty, as if there exists any valid complete plan $\pi$, there exists a valid complete plan $\pi' \in  \Pi[U_n,S]^*$ with $\pi \approx \pi'$.
\end{itemize}

Algorithm \ref{alg:generic1} gives the details on how to construct  $\Pi[U_i,T]^*$ for each $i$ and $T$.

The proof of correctness of Algorithm \ref{alg:generic1} proceeds by induction. 
Observe
first that for the case of $\Pi[U_0, T]^*$,  if $T \neq \emptyset$ then there is no possible plan in $\Pi[U_0, T]$, 
and so  we set $\Pi[U_0, T]^* = \emptyset$.
If $T = \emptyset$ then the only possible plan is the empty plan $\emptyset \rightarrow \emptyset)$. This plan is added to $\Pi[U_0,\emptyset]^*$, as it is trivially valid.
Thus $\Pi[U_0,T]^*$ is a $\Pi[U_0,T]$-\repset{} for each $T$.

So now assume that for all $T \subseteq S$ the set $\Pi[U_i, T]^*$ has been constructed and is a $\Pi[U_i, T]$-\repset{}. 
Now consider the construction of $\Pi[U_{i+1}, T]^*$ for some $T \subseteq S$.
It is clear that for any $\pi$ added to $\Pi[U_{i+1}, T]^*$, $\pi \in \Pi[U_{i+1}, T]$, and $\pi$ is eligible. Furthermore $\pi$ is authorized, as it is the union of the authorized plans $\pi' \in \Pi[U_i, T']$ and $(T'' \rightarrow u_{i+1})$. Thus every plan in $\Pi[U_{i+1}, T]^*$ is a valid plan in $\Pi[U_{i+1}, T]$.
On the other hand, suppose $\pi$ is a valid plan in $\Pi[U_{i+1}, T]$. Then let $T'' = \pi^{-1}(\{u_{i+1}\})$ and $T'=T \setminus T''$, and let $\pi'=\pi|_{U_i}$, so that $\pi = \pi' \cup (T'' \rightarrow u_{i+1})$. By assumption, there exists ${\pi'}^* \in \Pi[U_i, T]^*$ such that ${\pi'}^* \approx \pi'$. Consider the plan $\pi^* = {\pi'}^* \cup (T'' \rightarrow u_{i+1})$. It is clear that $\pi^*$ will be considered during the algorithm. Furthermore, as ${\pi'}^* \approx \pi'$ and $\pi = \pi' \cup (T'' \rightarrow u_{i+1})$, we have that $\pi^* \approx \pi$. Therefore $\pi^*$ is eligible (as $\pi$ is eligible) and also authorized (as it is the union of two authorized plans). Therefore $\pi^*$ is valid and will be added to $\Pi[U_{i+1}, T]^*$ unless $\Pi[U_{i+1}, T]^*$ already contains another plan $\approx$-equivalent to $\pi$.
Thus, $\Pi[U_{i+1}, T]^*$ contains a plan $\approx$-equivalent to $\pi$, from which it follows that $\Pi[U_{i+1}, T]^*$ is a $\Pi[U_{i+1}, T]$-\repset{}, as required.

It remains to analyse the running time of the algorithm.
By Proposition \ref{ass:settimes}, testing whether a pattern $p$ is in $\pat[U_i,T]^*$ and inserting $p$ into $\pat[U_i,T]^*$ takes $O^*(\log(|\pat[U_i,T]^*|))$ time.
Since by Assumption \ref{ass:polycheckequiv} and our assumption on the time to check constraints and authorizations
it takes polynomial time to check eligibility, authorization and $\approx$-equivalence of plans,
the running time of the algorithm is $O^*(\sum_{i=0}^{n-1}\sum_{T \subseteq S} \sum_{T' \subseteq T} \sum_{\pi \in \Pi[U_i,T']^*} \log(|\Pi[U_{i+1},T]^*|))$. It is clear by construction that $\Pi[U_i,T]^*$ contains at most one plan for each $\approx$-equivalence class over $\Pi[U_i,T]$, and so by definition $|\Pi[U_i,T]^*|\le w$ for all $i,T$. It follows that the running time of the algorithm is  $O^*(\sum_{i=0}^{n-1}\sum_{T \subseteq S} \sum_{T' \subseteq T} w\log w) = O^*(3^k  w\log w)$.
\end{proof}

\add{
\begin{remark}
 Rather than checking whether $\Pi[U_n,S]^*$ is non-empty at the end of the algorithm, we could instead check whether $\Pi[U_i,S]^*$ is empty after the construction of $\Pi[U_i,S]^*$ for each $i$. That is, we can stop our search as soon as we have a valid plan with task set $S$. This is likely to lead to saving in the running time of an implementation of the algorithm. 
 As this paper is concerned with the worst-case running time, which would be unaffected by this change, we perform the check at the end of the algorithm in the interest of clarity.
\end{remark}
}

\subsection{Application to User-Independent Constraints and its Optimality}


In this subsection, we show that  {\sc WSP} with user-independent  constraints is FPT.
Let $B_k$ denote the $k$th Bell number, the number of partitions of a set with $k$ elements. It is well-known that $B_k=2^{k\log k(1-o(1))}$ \cite{BeTa10}. 

\begin{lemma}\label{lem:uiwidth}
Let $u_1, \dots, u_n$ be any ordering of $U$, and 
let $\approx_{ui}$ be the plan-\indis{} relation given in Lemma \ref{lem:uirelation}.
Then $\approx_{ui}$ has \width{} $B_k$ with respect to $u_1, \dots, u_n$.
\end{lemma}
\begin{proof}
 For any plan $\pi$, the set $\{\pi^{-1}(u): u \in \user(\pi)\}$ is a partition of the tasks in $\task(\pi)$. 
 Furthermore, two plans that generate the same partition are equivalent under $\approx_{ui}$. Therefore the number of equivalence classes of $\approx_{ui}$ over $\Pi[U_i,T]$ is exactly the number of possible partitions of $T$, which is $B_{|T|}$. Thus, $B_k$ is the required \width.
\end{proof}

\begin{theorem}\label{thm:ui}
If all constraints are user-independent, then {\sc WSP} can be solved in time $O^*(2^{k\log k})$.
 \end{theorem}
\begin{proof}
Let $u_1, \dots, u_n$ be any ordering of $U$, and 
let $\approx_{ui}$ be the plan-\indis{} relation given in Lemma \ref{lem:uirelation}.

By Lemma \ref{lem:uiwidth}, $\approx_{ui}$ has \width{} $B_k$ with respect to $u_1, \dots, u_n$.
 Furthermore, by Corollary \ref{cor:uiencoding}, there exists an encoding for $\approx_{ui}$.
 Therefore, we may apply Theorem  \ref{thm:linear} with $w=B_k$, to get an algorithm with running time 
 $O^*(3^k  B_k\log(B_k)) 
 = O^*(3^k 2^{k\log k(1-o(1))} \log(2^{k\log k(1-o(1))}))
 = O^*(2^{k \log k})$.
\end{proof}


The running time $O^*(2^{k\log k})$ obtained is optimal in the sense that no algorithm of running time $O^*(2^{o(k\log k)})$ exists, unless the ETH fails.
In the proof of the following theorem, we use a result of Lokshtanov et al. (Theorem 2.2, \cite{LoMaSa}).

\begin{theorem}\label{thm:LB}
There is no algorithm for {\sc WSP} with user-independent constraints of
running time $O^*(2^{o(k \log k)})$, unless the ETH fails.
\end{theorem}

\begin{proof}

We give a reduction from the problem {\sc $k \times k$ Independent Set}:
Given an integer parameter $k$ and a graph $G$ with vertex set $V =
\{(i,j):\ i, j \in [k]\}$, decide whether
$G$ has an independent set $I$ such that $|I|=k$ and for each $r \in [k]$,
there exists $i$ such that $(r,i) \in I$.

Informally, {\sc $k \times k$ Independent Set} gives us a graph on a $k\times k$ grid
of vertices, and asks whether there is an independent set with one vertex
from each row.
Lokshtanov et al. \cite{LoMaSa} proved that there is no algorithm to solve
{\sc $k \times k$ Independent Set} in time $2^{o(k \log k)}$, unless the ETH
fails.

Consider an instance of {\sc $k \times k$ Independent Set} with graph $G$. We
will first produce an equivalent instance of {\sc WSP} in which the
constraints are not user-independent. We will then refine this instance to
one with user-independent constraints.

Let $U = \{u_1, \dots, u_k\}$ be a set of $k$ users and $S = \{s_1, \dots,
s_k\}$ a set of $k$ tasks. Let the authorization lists be $A(s_i) = U$ for all
$i \in [k]$.
For $i,j,h,l \in [k]$, let $c((i,j),(h,l))$ denote the constraint with
scope $\{s_i,s_h\}$, and which is satisfied by any plan $\pi$
unless $\pi(s_i)=u_j$ and $\pi(s_h)=u_l$. 
For every pair of vertices $(i,j),(h,l)$ which are adjacent in $G$, add
the constraint $c((i,j),(h, l))$ to $C$.


We now show that $(S,U,A,C)$ is a {\sc Yes}-instance of {\sc WSP} if and
only if $G$ has an independent set with one vertex from each row.
Suppose $(S,U,A,C)$ is a {\sc Yes}-instance of {\sc WSP} and let $\pi$ be
a valid complete plan. Then for each $i \in [k]$, let $f(i)$ be the unique
$j$ such that $\pi(s_i)=u_j$. Then $I = \{(i,f(i)): i \in [k]\}$ is a set
with one vertex from each row in $G$; furthermore, as $\pi$ satisfies
every constraint, no edge in $G$ contains more than one element of $I$,
and so $I$ is an independent set.

Conversely, suppose $G$ is a {\sc Yes}-instance of {\sc $k \times k$ Independent Set}.
For each $i \in [k]$, let $f(i)$ be an integer such that $(i, f(i)) \in
I$. Then observe that $\bigcup_{i=1}^{k} (\{s_i\} \rightarrow u_{f(i)})$
is a valid complete plan.

We now show how to reduce $(S,U,A,C)$ to an instance of {\sc WSP} in which
all constraints are user-independent. The main idea is to introduce some
new tasks representing the users, and in the constraints, replace the
mention of a particular user with the mention of the user that performs a
particular task.


Create $k$ new tasks $t_1, \dots, t_k$ and let $S' = S \cup \{t_1,, \dots,
t_k\}$. Let the authorization lists be $A'(s)=U$ for each $s\in S$ and $A'(t_i)=\{u_i\}$ for each $i\in [k].$
For each constraint  $c((i,j),(h,l))$ in $C$, let $d((i,j),(h,l))$ be the
constraint with scope $ \{s_i,s_h,t_j,t_l\}$, which is
satisfied by any plan $\pi$ unless $\pi(s_i)=\pi(t_j)$ and
$\pi(s_h)=\pi(t_l)$. Let initially $C'=C$. Now replace, in $C'$, 
every constraint $c((i,j),(h,l))$
with $d((i,j,),(h,l))$.


Since they are defined by equalities, and no users are mentioned, the constraints in $C'$ are user-independent.
We now show that $(S',U,A',C')$ is equivalent to $(S,U,A,C)$.
First, suppose that $\pi$ is a valid complete plan for $(S,U,A,C)$. Then
let $\pi': S' \rightarrow U$ be the plan such that $\pi'(s_i)=\pi(s_i)$
for all $i \in [k]$, and $\pi'(t_j) = u_j$ for all $j \in [k]$. It is easy
to check that if $\pi$ satisfies every constraint of $C$ then $\pi'$
satisfies every constraint of $C'$. Since $\pi'$ is an authorized and eligible plan, $\pi'$ is a valid complete plan for $(S',U,A',C').$

Conversely, suppose that $\pi'$ is a valid complete plan for
$(S',U,A',C')$.  Since $A'(t_i)=\{u_i\}$ for each $i\in [k],$ $\pi'(t_i)=u_i$ for every $i\in [k].$ 
For each $i \in [k]$, let $f(i)$ be the unique integer such that
$\pi'(s_i)=u_{f(i)}$. Then define $\pi:S \rightarrow U$ by $\pi(s_i)
= u_{f(i)}$, and observe that all constraints in $C$ are satisfied by
$\pi$. So, $\pi$ is a valid complete plan for $(S,U,A,C)$.
\end{proof}

\subsection{Application to Equivalence Relation Constraints}

It is known that restricting the WSP to have only equivalence relation constraints is enough to ensure FPT~\cite{CrGuYeJournal}.  However, we can derive this result by applying our algorithm directly having shown the appropriate properties of the language of equivalence relation constraints.  This serves to demonstrate the wide applicability of our approach.

\begin{lemma}\label{lem:equivwidth}
 Let $\approx_{e}$ be the plan-indistinguishability relation given for a set of equivalence relation constraints in Lemma \ref{lem:equivRelation}. Then there exists an ordering $u_1, \dots, u_n$ of $U$ such that $\approx_{e}$ has \width{} $2^k$ with respect to $U$.
\end{lemma}

\begin{proof}
Suppose $\sim$ is an equivalence relation on users, and let $V_1, \dots, V_p$ be the equivalence classes of $\sim$ over $U$.
Suppose all constraints are of the form $(s_i,s_j, \sim)$ or $(s_i,s_j, \not\sim)$.

Let $u_1, \dots, u_n$ be an ordering of $U$ such that all the elements of $V_j$ appear before all the elements of $V_{j'}$, for any $j < j'$. Thus, for any $i$ and any plan $\pi$ with $\user(\pi)=U_i=\{u_1, \dots, u_i\}$, there is at most one integer $j_i$ such that 
$V_{j_i} \cap \user(\pi) \neq \emptyset, V_{j_i} \setminus \user(\pi) \neq \emptyset$.

It follows that any two plans $\pi_1, \pi_2 \in \Pi[U_i,T]$ are $\approx_{e}$-equivalent, for any $i \in [n], T \subseteq S$, 
provided that $\pi_1(t) \in V_{j_i}$ if and only if $\pi_2(t) \in V_{j_i}$ for any $t \in T$. 
Therefore $\approx_{e}$ has at most $2^k$ equivalence classes over $\Pi[U_i,T]$, as required.
\end{proof}

\begin{theorem}\label{thm:equiv}
 Suppose $\sim$ is an equivalence relation on $U$.
Suppose all constraints are of the form $(s_i,s_j, \sim)$ or $(s_i,s_j, \not\sim)$.
Then {\sc WSP} can be solved in time $O^*(6^k)$.
 \end{theorem}
\begin{proof}
Let $u_1, \dots, u_n$ be the ordering of $U$ given by Lemma \ref{lem:equivwidth}, and 
let $\approx_{e}$ be the plan-\indis{} relation given in Lemma \ref{lem:equivRelation}.

By Lemma \ref{lem:equivwidth}, $\approx_{e}$ has \width{} $2^k$ with respect to $u_1, \dots, u_n$.
 Furthermore by Corollary \ref{cor:nestedEquivEncodingencoding}, there exists an encoding for $\approx_{e}$.
 Therefore, we may apply Theorem  \ref{thm:linear} with $w=2^k$, to get an algorithm with running time 
 $O^*(3^k  2^k\log(2^k)) 
  = O^*(6^k)$.
\end{proof}

\section{Unions of Constraint Languages}\label{sec5}

In this section we show how our approach allows us easily to combine constraint languages shown to be FPT for the WSP.  We do not need to build bespoke algorithms for the new constraint language obtained, only to show that the two languages are in some sense compatible.

This highlights the advantages of our 
approach over previous methods, which required the development of new 
algorithms when different constraint languages were combined in an 
instance of WSP (see~\cite{CrGuYeJournal}, for example).


\begin{theorem}\label{thm:compatible}
 Let $(S,U,A,C_1 \cup C_2)$ be an instance of {\sc WSP}, and suppose $\approx_1$ is a plan-\indis{} relation with respect to $C_1$ and $\approx_2$ is a plan-\indis{} relation with respect to $C_2$.
 Given an ordering $u_1, \dots, u_n$ of $U$, let $W_1$ be the \width{} of $\approx_1$ with respect to $u_1, \dots, u_n$ and $W_2$ the \width{} of $\approx_2$ with respect to $u_1, \dots, u_n$.
 
 Let $\approx$ be the equivalence relation such that $\pi \approx \pi'$ if and only if $\pi \approx_1 \pi'$ and $\pi \approx_2 \pi'$. Then $\approx$ is a plan-\indis{} relation with respect to $C_1 \cup C_2$, and $\approx$ has \width{} $W_1W_2$ with respect to $u_1, \dots, u_n$.
 
\end{theorem}

\begin{proof}
 
 We first show that $\approx$ is a plan-\indis{} relation with respect to $C_1 \cup C_2$. Let $\pi$ and $\pi'$ be eligible plans (with respect to $C_1 \cup C_2$).
 As $\pi \approx \pi'$ implies $\pi \approx_1 \pi'$ and $\approx_1$ satisfies the conditions of a plan-\indis{} relation, it is clear that if $\pi \approx \pi'$ then $\user(\pi)=\user(\pi')$ and $\task(\pi)=\task(\pi')$. Now consider a plan $\pi''$ disjoint from $\pi$ and $\pi'$. As $\approx_1$ is a plan-\indis{} relation with respect to $C_1$ and $\pi \approx_1 \pi'$, we have that $\pi \cup \pi''$ is $C_1$-eligible if and only if $\pi' \cup \pi''$ is.
 Similarly $\pi \cup \pi''$ is $C_2$-eligible if and only if $\pi' \cup \pi''$ is. Observing that a plan is $C_1 \cup C_2$-eligible if and only if it is $C_1$-eligible and $C_2$-eligible, this implies that $\pi \cup \pi''$ is $C_1 \cup C_2$-eligible if and only if $\pi' \cup \pi''$ is. 
 Thus we have that $\pi$ and $\pi'$ are extension equivalent. 
 
 As $\approx_1$ and $\approx_2$ are plan-\indis{} relations, we have that $\pi \cup \pi'' \approx_1 \pi' \cup \pi''$ and $\pi \cup \pi'' \approx_2 \pi' \cup \pi''$, and therefore $\pi \cup \pi'' \approx \pi' \cup \pi''$.
 Thus, $\approx$ satisfies all the conditions of a plan-\indis{} relation.

 To bound the \width{} of $\approx$ with respect to $u_1, \dots, u_n$, consider any $T \subseteq S$ and $U_i  = \{u_1, \dots, u_i\}$. It is enough to note that any $\approx$-equivalent plans in $\Pi[U_i,T]$ must be in the same $\approx_1$ and $\approx_2$-equivalence classes. As there are at most $W_1$ choices for the $\approx_1$-equivalence class and at most $W_2$ choices for the $\approx_2$ equivalence class, $\approx$ has at most $W_1W_2$ equivalence classes over $\Pi[U_i,T]$.
\end{proof}

\begin{remark}\label{rem:composeEncodings}
Given an encoding $\enc_1$ for $\approx_1$ and an encoding $\enc_2$ for $\approx_2$, we may construct an encoding for $\approx$.
Given a plan $\pi$, let $\enc(\pi)$ be the ordered pair $(\enc_1(\pi), \enc_2(\pi))$.
 It is clear that $\enc(\pi) = \enc(\pi')$ if and only if $\pi \approx \pi'$.
 
 Given sets $T \subseteq S$ and $U_i  = \{u_1, \dots, u_i\}$, fix linear orderings of $\enc_1(\Pi[U_i,T])$ and $\enc_2(\Pi[U_i,T])$. Then let $\preceq$ be the lexicographic ordering of $\enc(\Pi[U_i,T]) = \enc_1(\Pi[U_i,T]) \times \enc_2(\Pi[U_i,T]).$
 \end{remark}

\add{There is nothing to stop us applying Theorem \ref{thm:compatible} multiple times, in order to get a plan-\indis{} relation with bounded \width{} for a union of several constraint languages. Note that the \width{} can be expected to grow exponentially with the number of languages in the union. Thus, it makes sense to only apply Theorem \ref{thm:compatible} to the union of a small number of languages. However, as long as there is a fixed number of languages, and  each has a plan-\indis{} relation with fixed-parameter \width{}, the resulting union of languages will also have a plan-\indis{} relation with fixed-parameter \width{}.}

We  can now use this result directly to show that if all our constraints are either user independent or equivalence relation constraints, then the WSP is still FPT.

\begin{theorem}
 Suppose $\sim$ is an equivalence relation on $U$.
 Let $(S,U,A,C)$ be an instance of {\sc WSP}, and suppose that all  constraints are  either of the form, $(s_1, s_2, \sim)$, $(s_1, s_2, \not\sim)$ or user-independent constraints.
Then {\sc WSP} can be solved in time $ O^*(2^{k \log k + k})$.
\end{theorem}

%

\begin{proof}
 Let $C_{e} \subseteq C$ be the set of constraints of the form $(s_1, s_2, \sim)$, $(s_1, s_2, \not\sim)$, and let $C_{ui}$ be the remaining (user-independent) constraints.
 
Let $u_1, \dots, u_n$ be the ordering of $U$ given by Lemma \ref{lem:equivwidth}.
By Lemmas \ref{lem:equivRelation} and \ref{lem:equivwidth}, there exists a plan-\indis{} relation $\approx_{e}$ for $C_{e}$ that has \width{} $2^k$ with respect to $u_1, \dots, u_n$.
Furthermore by Corollary \ref{cor:nestedEquivEncodingencoding}, $\approx_{e}$ has an encoding.
By Lemmas \ref{lem:uirelation} and \ref{lem:uiwidth}, there exists a plan-\indis{} relation $\approx_{ui}$ for $C_{ui}$ that has \width{} $B_k$ with respect to $u_1, \dots, u_n$.
Furthermore by Corollary \ref{cor:uiencoding}, $\approx_{ui}$ has an encoding.

Therefore by Theorem \ref{thm:compatible}, we may find a plan-\indis{} relation $\approx$ for $C$, such that $\approx$ has \width{} $B_k \cdot 2^k$ with respect to $u_1, \dots, u_n$ and $\approx$ has an encoding. Thus we may apply Theorem \ref{thm:linear} with $w = B_k \cdot 2^k$, 
to get a running time of 
 $O^*(3^k  B_k2^k\log(B_k2^k)) 
 = O^*(3^k 2^{k\log k(1-o(1)) + k} \log(2^{k\log k(1-o(1)) + k}))
 = O^*(2^{k \log k + k})$.
\end{proof}

\section{Conclusion} \label{sec6}
In this paper we introduced
\remove{the first generic algorithm for WSP} 
\add{an algorithm applicable to a wide range of WSP problems, based on the notion of plan-\indis,}
and showed that our \remove{generic} algorithm is powerful enough
to be optimal, in a sense, for the wide class of user-independent
constraints. The generic algorithm is also a fixed-parameter algorithm for
equivalence relation constraints, which are not user-independent. We show how to deal
with unions of different types of constraints using our generic algorithm.
In particular, we prove that the generic algorithm is a fixed-parameter
algorithm for the union of user-independent and equivalence relation constraints.

Due to the difficulty of acquiring real-world workflow instances, Wang and
Li \cite{WangLi10} used synthetic data in their experimental study. Wang and Li
encoded WSP as a pseudo-Boolean SAT problem in order to use
a pseudo-Boolean SAT solver SAT4J to solve several instances of WSP. We
have implemented our algorithm and compared its performance to SAT4J on
another set of synthetic instances of WSP in \cite{CCGGJ14}.
 These instances use $k = 16$, $20$ and $24$, $n = 10k$ and
user-independent constraints of three different types: we vary the number
of constraints and the proportions of different constraints types; each
user is authorized for between $1$ and $8$ tasks for $k=16$, between 1 and
10 tasks for $k=20$, and between 1 and 12 tasks for $k=24$.  The algorithm
was implemented in C++ and has been enhanced by the inclusion of techniques
employed in CSP solving, such as propagation.  We also converted WSP
instances into pseudo-Boolean problems for processing by SAT4J. All
experiments were performed on a MacBook Pro computer having a 2.6 GHz Intel
Core i5 processor and 8 GB 1600 MHz DDR3 RAM (running Mac OS X 10.9.2).

For lightly-constrained instances, SAT4J was often faster than our
algorithm, largely because the number of patterns considered by our
algorithm is large for such instances.  However, for highly-constrained
instances, SAT4J was unable to compute a decision for a number of instances
(because it was running out of memory), in sharp contrast to our algorithm
which solved all instances.  Overall, on average, our algorithm was faster
than SAT4J and, in particular, was two orders of magnitude faster for $k =
16$.  Moreover, the time taken by our algorithm varies much less than that of
SAT4J, even for unsatisfiable instances, because the time taken is
proportional to the product of the number of patterns and the number of
users.  (In particular, it is much less dependent on the number of
constraints, 
a parameter that
can cause
significant fluctuations in the time taken by SAT4J because this leads to a
sharp increase in the number of variables in the pseudo-Boolean encoding.)
 Full details of our results may be found in~\cite{CCGGJ14}.

\paragraph{Acknowledgement.} Our research was supported by EPSRC grant
EP/K005162/1

\newpage


\begin{thebibliography}{10}
 \bibitem{ansi-rbac04}
American National Standards Institute.
\newblock {\em {ANSI INCITS} 359-2004 for Role Based Access Control}, 2004.

\bibitem{BaKo}
L. Barto and M. Kozik.
\newblock Constraint satisfaction problems of bounded width.
\newblock In {\em FOCS}, pages 595--603, 2009.

\bibitem{BaBuKa}
D.A.~Basin, S.J.~Burri, and G.~Karjoth.
\newblock Obstruction-Free Authorization Enforcement: Aligning Security with Business Objectives.
\newblock In {\em CSF 2011}, pages 99--113, 2011.

\bibitem{catalogue}
N.~Beldiceanu, M.~Carlsson, and J.X. Rampon.
\newblock Global constraints catalogue, Sept 2012.

\bibitem{BeTa10}
D.~Berend and T.~Tassa.
\newblock Improved bounds on {B}ell numbers and on moments of sums of random
  variables.
\newblock {\em Probability and Math. Stat.}, 30:185--205, 2010.

\bibitem{BeFeAt99}
E.~Bertino, E.~Ferrari, and V.~Atluri.
\newblock The specification and enforcement of authorization constraints in
  workflow management systems.
\newblock {\em ACM Trans. Inf. Syst. Secur.}, 2(1):65--104, 1999.

\bibitem{BeboFe01}
E. Bertino, P.A. Bonatti, and E. Ferrari.
\newblock {TRBAC}: A temporal role-based access control model.
\newblock {\em ACM Trans. Inf. Syst. Secur.}, 4(3):191--233, 2001.


\bibitem{BoCyKrNe13}
H.L. Bodlaender, M. Cygan, S. Kratsch, and J. Nederlof.
\newblock Solving weighted and counting variants of connectivity problems
  parameterised by treewidth deterministically in single exponential time.
\newblock In {\em ICALP}, 2013, to appear.

\bibitem{BuDa}
A.~Bulatov and V.~Dalmau.
\newblock A simple algorithm for {M}al'tsev constraints.
\newblock {\em SIAM Journal on Computing}, 36(1):16--27, 2006.

\bibitem{CCGGJ14}
D.~Cohen, J.~Crampton, A.~Gagarin, G.~Gutin, and M.~Jones.
\newblock Engineering algorithms for workflow satisfiability problem with
user-independent constraints.
\newblock In {\em FAW}, 2014, to appear.

\bibitem{CLRS01}
T.H.~Cormen, C.E.~Leiserson, R.L.~Rivest, and C.~Stein. 
\newblock Introduction to Algorithms, Second Edition. MIT Press and McGraw-Hill, 2001. 


\bibitem{CCGJR}
J.~Crampton, R.~Crowston, G.~Gutin, M.~Jones, and M.S. Ramanujan.
\newblock Fixed-parameter tractability of workflow satisfiability in the
  presence of seniority constraints.
\newblock In {\em FAW-AAIM}, 198-209 2013.


\bibitem{CrGu}
J.~Crampton and G.~Gutin.
\newblock Constraint expressions and workflow satisfiability.
\newblock In {\em SACMAT}, 73-84, 2013.

\bibitem{CrGuYeJournal}
J.~Crampton, G.~Gutin, and A.~Yeo.
\newblock On the parameterized complexity and kernelization of the workflow
  satisfiability problem.
\newblock {\em ACM Trans. Inform. System \& Secur.}, 16(1):4, 2013.
\newblock Preliminary version in ACM Conf. Comput. \& Communic. Secur. 2012,
857--868.

\bibitem{Cr05}
J. Crampton.
\newblock A reference monitor for workflow systems with constrained task
  execution.
\newblock In Elena Ferrari and Gail-Joon Ahn, editors, {\em SACMAT}, pages
  38--47. ACM, 2005.
  
  \bibitem{DechterCP} R. Dechter, Constraint Processing, Morgan Kaufmann, 2003.

\bibitem{DowFel99}
R.G. Downey and M.R. Fellows.
\newblock {\em Parameterized complexity}.
\newblock Springer, 1999.

\bibitem{FoLoSa}
F.V. Fomin, D.~Lokshtanov, and S.~Saurabh.
\newblock Efficient computation of representative sets with applications in
  parameterized and exact algorithms.
\newblock {\em CoRR}, abs/1304.4626, 2013.

\bibitem{ImPaZa01}
R. Impagliazzo, R. Paturi, and F. Zane.
\newblock Which problems have strongly exponential complexity?
\newblock {\em J. Comput. Syst. Sci.}, 63(4):512--530, 2001.

\bibitem{JoBeLaGh05}
J. Joshi, E. Bertino, U. Latif, and A. Ghafoor.
\newblock A generalized temporal role-based access control model.
\newblock {\em IEEE Trans. Knowl. Data Eng.}, 17(1):4--23, 2005.

\bibitem{KrBuJe}
A.~Krokhin, A.~Bulatov, and P.~Jeavons.
\newblock The complexity of constraint satisfaction: an algebraic approach.
\newblock In 
  {\em Structural Theory of Automata, Semigroups, and Universal
  Algebra}, volume 207 of {\em NATO Science Series II: Mathematics, Physics and
  Chemistry}, pages 181--213. Springer, 2005.
  
  
\bibitem{leBePa10}
D.~Le Berre, and A~Parrain.
\newblock The {SAT4J} library, release 2.2,
\newblock {\em J. Satisf. Bool. Model. Comput.}, 7: 59--64, 2010.
  
\bibitem{LoMaSa}
D.~Lokshtanov, D.~Marx, and S.~Saurabh.
\newblock Slightly superexponential parameterized problems.
\newblock In {\em SODA}, pages 760--776, 2011.


\bibitem{Nie06}
R.~Niedermeier.
\newblock {\em Invitation to Fixed Parameter Algorithms}.
Oxford
  U. Press, 2006.
  
\bibitem{CPhandbook}  F. Rossi, P. van Beek and T. Walsh, Handbook of Constraint Programming (Foundations of Artificial Intelligence), Elsevier, 2006.


\bibitem{SaCoFeYo96}
R.S. Sandhu, E.J. Coyne, H.L. Feinstein, and C.E. Youman.
\newblock Role-based access control models.
\newblock {\em IEEE Computer}, 29(2):38--47, 1996.


\bibitem{WangLi10}
Q. Wang and N. Li.
\newblock Satisfiability and resiliency in workflow authorization systems.
\newblock {\em ACM Trans. Inf. Syst. Secur.}, 13(4):40, 2010.



\end{thebibliography}
\end{document}